\documentclass[prl,nobalancelastpage,twocolumn,superscriptaddress,nolongbibliography]{revtex4-2}
\pdfoutput=1

\usepackage{color,amsthm,amsmath,amsxtra,amsfonts,dsfont,graphicx,bm,amssymb}
\usepackage[colorlinks=true,linkcolor=blue, citecolor=blue, urlcolor=blue, bookmarks]{hyperref}
\usepackage{centernot}
\usepackage[dvipsnames]{xcolor}
\usepackage{graphicx}
\usepackage{tikz}
\usepackage{braket}
\usepackage{multirow}

\newcommand{\be}{\begin{equation}}
	\newcommand{\ee}{\end{equation}}
\newcommand{\bea}{\begin{eqnarray}}
	\newcommand{\eea}{\end{eqnarray}}
\newcommand{\bse}{\begin{subequations}}
	\newcommand{\ese}{\end{subequations}}

\theoremstyle{plain}

\newtheorem{prop}{Proposition}
\newtheorem{lem}{Lemma}
\newtheorem{rslt}{Result}

\newcommand{\prlsection}[1]{{\em {#1}.---~}}

\begin{document}
	
	\title{Approximating many-body quantum states with quantum circuits and measurements}
	
	\begin{abstract}
		We introduce protocols to prepare many-body quantum states with quantum circuits assisted by local operations and classical communication. We show that by lifting the requirement of exact preparation, one can substantially save resources. In particular, the so-called $W$ and, more generally, Dicke states require a circuit depth and number of ancillas per site that are independent of the system size. As a byproduct of our work, we introduce an efficient scheme to implement certain non-local, non-Clifford unitary operators. We also discuss how similar ideas may be applied in the preparation of eigenstates of well-known spin models, both free and interacting.
	\end{abstract}

	\author{Lorenzo \surname{Piroli}}
	\affiliation{Dipartimento di Fisica e Astronomia, Università di Bologna and INFN, Sezione di Bologna, via Irnerio 46, I-40126 Bologna, Italy}
	
	\author{Georgios \surname{Styliaris}}
	\affiliation{Max-Planck-Institut f{\"{u}}r Quantenoptik, Hans-Kopfermann-Str. 1, 85748 Garching, Germany}
	\affiliation{Munich Center for Quantum Science and Technology (MCQST), Schellingstr. 4, D-80799 München, Germany}

	\author{J.~Ignacio \surname{Cirac}}
	\affiliation{Max-Planck-Institut f{\"{u}}r Quantenoptik, Hans-Kopfermann-Str. 1, 85748 Garching, Germany}
	\affiliation{Munich Center for Quantum Science and Technology (MCQST), Schellingstr. 4, D-80799 München, Germany}

	\maketitle
	
	
	\prlsection{Introduction} The preparation of many-body quantum states plays a pivotal role in quantum simulation~\cite{cirac2012goals}. On the one hand, some of those states are required to exploit the field of quantum sensing~\cite{degen2017quantum}, quantum communication~\cite{gisin2007quantum}, or play a crucial role in quantum information theory~\cite{horodecki2009quantum}. On the other, they allow to investigate quantum many-body systems, extracting properties that otherwise are difficult to compute. Furthermore, some of them can be useful to initialize quantum algorithms that prepare ground states~\cite{georgescu2014quantum,mcArdle2020quantum,bauer2020quantum} or thermal states~\cite{chowdhury2017quantum,lu2021algorithms,temme2011quantum,motta2020determining,cohn2020minimal}.
	
	As current noisy intermediate-scale quantum (NISQ) devices~\cite{preskill2018quantum} are limited in the number of qubits and the coherence time, it is very important to devise efficient preparation schemes making use of the minimum amount of resources. Following early ideas~\cite{briegel2001persistent,raussendorf2005long}, an emerging theme is that preparation protocols using unitary circuits can be improved by making use of additional ancillas, measurements, and feedforward operations, notably in the context of topological order~\cite{piroli2021quantum,verresen2021efficiently,lu2022measurement,bravyi2022adaptive,lee2022decoding,tantivasadakarn2023hierarchy,lootens2023low,tantivasadakarn2023shortest,zhu2023nishimori,smith2023deterministic,sukeno2023quantum,gunn2023phases,malz2024preparation,okuda2024anomaly}. These ingredients are very natural from the point of view of quantum information, where they are called local operations and classical communication (LOCC)~\cite{horodecki2009quantum}.
	
	\begin{table}[t]
		\begin{ruledtabular}
			\begin{tabular}{c |c  c   c c } 
				& Ref.   &$D$    & $N_a$    & $N_r$ \\
				\hline
				\multirow{4}{2.5em}{$W$}
				&  Result~\ref{thm_dicke}  & $O(\ln \ln 1/\varepsilon)$  & 1 & $O(1)$ \\
				&   Result~\ref{thm_dicke}  & $O(1)$  & $O(\ln \ln 1/\varepsilon)$ & $O(1)$ \\ 
				& Result~\ref{thm_W}  & $O(1)$  & $1$ & O($1/\sqrt{\varepsilon}$) \\ 
				&  Ref.~\cite{buhrman2023state}  & $O(1)$ & $O(\ln N)$ & $1$ \\ \hline
				\multirow{4}{2.5em}{Dicke} & Result~\ref{thm_dicke}&  $O(1)$ & $ O(\ell_{M,\varepsilon})$ & $O(\sqrt{M})$    \\
				& Result~\ref{thm_dicke}&  $O(\ell_{M,\varepsilon})$ & $ 1$ & $O(\sqrt{M})$    \\
				&  Result~\ref{thm_dicke_improved} &   $O(M^{1/4} \ell^2_{M,\varepsilon} )$ & $1 + \ell_{M,\varepsilon}/N$   & $1$  \\
				& Ref.~\cite{buhrman2023state} &  $O(1)$ & $O(N \ln N)$ & $1$
			\end{tabular}
			\caption{Summary of our results and comparison with previous work [$M$: number of excitations; $\varepsilon$: infidelity]. The resources are the depth $D$ (including LOCC, if applicable) the number of ancillas per site $N_a$ and of repetitions $N_r$. A trade-off is possible in some cases, and we give variants optimizing either $D$, $N_a$, or $N_r$ [$\ell_{M,\varepsilon}$ is defined in Eqs.~\eqref{eq:ell_def_0} and \eqref{eq:ell_def} for Results~\ref{thm_dicke} and~\ref{thm_dicke_improved}, respectively]. Ref.~\cite{buhrman2023state} allows for $M = O(\sqrt N)$ while, for arbitrary $M$, $N_a = O({\rm Poly} (N) )$, $D = O(\ln N)$.} \label{table:scalings}
		\end{ruledtabular}
	\end{table}
	
	The goal of this work is to introduce protocols that save additional resources as compared to existing schemes. As we show, this is achieved by relaxing the condition of preparing the states exactly and deterministically. This does not cause any disadvantage since for any realistic device exact preparation will never be possible. A cornerstone of our schemes is a non-local unitary operation that can be efficiently implemented and that, in contrast to those introduced in Ref.~\cite{piroli2021quantum}, is not Clifford~\cite{gottesman1997stabilizer,gottesman1998theory}. We also show how this operation can help to save resources by creating one-by-one excitations in spin systems. 
	
	In this letter, we identify as resources the depth $D$ of the quantum circuit (QC), the number of experimental repetitions $N_r$, and the number of ancillas per qubit $N_a$ needed in order to produce an infidelity $I=\varepsilon$. It is important to carefully define the depth of the circuit, which will be done later. We anticipate that, contrary to some of the protocols in Ref.~\cite{piroli2021quantum}, we will only allow for LOCC where all the measurements are executed in parallel. We also note that, in our schemes, one can trade among different resources, but we will be mostly concerned with saving $N_a$ and $D$, which are arguably more important for the first generation of quantum computers.
	
	Our main result is to show how to prepare the $N$-qubit \emph{Dicke states}~\cite{dicke1954coherence}
	\be\label{eq:def_dicke}
	|W(M)\rangle = Z^{-1}_M (S^{+})^M |0\ldots 0\rangle\,,
	\ee
	where $S^{\pm }=\sum_{m=1}^N \sigma^{\pm}_m$, $Z_M$ is a normalization factor, while $\sigma_m^{\pm }$ are the ladder operators at position $m$. The states~\eqref{eq:def_dicke} are eigenstates of the Dicke Hamiltonian $H_D = S^+ S^- + S^- S^+$, where $M\in [0,N]$ is the number of excitations. They were defined in the Dicke model of superradiance~\cite{dicke1954coherence, gross1982superradiance}, and are expected to be useful in different kinds of quantum simulations of that model, see e.g. Ref.~\cite{deVega2008matter}. Our interest in these states is two-fold. On the one hand, they play a fundamental role in quantum information science and, in particular, in metrology. As a consequence, a significant amount of experimental~\cite{kiesel2007experimental,wieczorek2009experimental,noguchi2012generation,hume2009preparation} and theoretical~\cite{wu2017generation,duan2003efficient,lamata2013deterministic,ionicioiu2008generalized} work has studied protocols for their preparation in digital and analog quantum platforms. On the other hand, Dicke states have resisted previous attempts to devise preparation schemes using finite-depth circuits and a finite number of ancillas per site~\cite{bartschi2019deterministic, piroli2021quantum,tantivasadakarn2023hierarchy}, raising the question of whether there are some fundamental limitations to achieve this task.
	
	The preparation of Dicke states with LOCC has been previously considered in the literature. In Ref.~\cite{piroli2021quantum} a protocol was proposed to prepare the $W$ state that uses a QC with $D=O(1)$, but requires sequential use of LOCC, \emph{i.e.} a $O(N)$ preparation time. In Ref.~\cite{buhrman2023state} an ingenious approach was introduced to deterministically prepare the $W$ and Dicke states with constant depth but $N_a$ scaling with $N$. Instead, the protocols developed in this work allow, for any fixed desired infidelity and a constant number of excitations, $N$-independent resources (Table~\ref{table:scalings}). Our approach is very different from that of Ref.~\cite{buhrman2023state} and arguably simpler.   The physical intuition behind our protocol  is that  the Dicke state may be obtained by measuring the total number of excitations, starting from some suitable unentangled (and thus, easily prepared) initial state.  This strategy is very natural, as it relies on the interpretation of Dicke states as made of quasiparticle excitations. In fact, we note that a similar idea has been first followed in Ref.~\cite{wang2021preparing}, in a very different analog setting. In our work we solve the non-trivial problem of implementing this idea using finite-depth circuits and LOCC.
	
	We also discuss how similar ideas may be useful to prepare certain states of interest in many-body physics. We consider the eigenstates of the XX Hamiltonian and present a deterministic preparation protocol with $D=O(MN)$, where $M$ is the number of excitations. While our protocol is less efficient than the state-of-the-art unitary algorithm requiring $O(N)$ depth~\cite{kivlichan2018quantum,zhang2018quantum,arute2020observation}, our method is of interest as it is in principle applicable to more general states and could lead to further improvement or generalizations. Finally, we also discuss how extensions of our ideas may allow one to prepare eigenstates of interacting spin chains, including the so-called Richardson-Gaudin model~\cite{richardson1963restricted,richardson1964exact}.

	\prlsection{Non-Clifford unitaries from QCs and LOCC} We consider $N$ qubits in one spatial dimension. The associated Hilbert space is $\mathcal{H}=\mathcal{H}_{2}^{\otimes N}$, with $\mathcal{H}_{2}\simeq \mathbb{C}^2$, while we denote by $\{|0\rangle,|1\rangle\}$ the computational basis. We attach to each qubit $N_a$ ancillas. Then, we define the local QCs as the unitaries $ W =  W_{\ell} \ldots W_2 W_1$, where each ``layer'' $W_{n}$ contains quantum gates acting on disjoint pairs of nearest-neighbor qubits and possibly the associated ancillas. In between each layer, we allow for LOCC consisting of a round of measurements executed in parallel, classical processing of the outcomes and local corrections (executed in parallel). We define the circuit depth as the total number of unitary layers and LOCC steps. 
	
	We begin by showing how  to implement non-Clifford operations of the form
	\be
	\label{eq:unitary}
	V = |0\rangle_b\langle 0| \otimes U^{(0)} + |1\rangle_b\langle 1| \otimes U^{(1)}\,,
	\ee
	where $U^{(k)}=\otimes_{j=1}^N U_{k,j}$ and $U_{k,j}$ act on system qubit $j$, with $k=0,1$. Here, $\ket{\pm}=(\ket{0}\pm \ket{1})/\sqrt{2}$, while $b$ is the ancilla placed at position $1$. The form \eqref{eq:unitary} includes quantum fan-out gates, which are useful in quantum computing~\cite{hoyer2005quantum,takahashi2016collapse,buhrman2023state}. We prove the following:
	\begin{rslt}\label{lem_v}
		$V$ can be implemented deterministically (i.e. by a single repetition $N_r=1$), using $N_a=1$ and $D=6$.
	\end{rslt}
	\begin{figure}
		\includegraphics[scale=0.42]{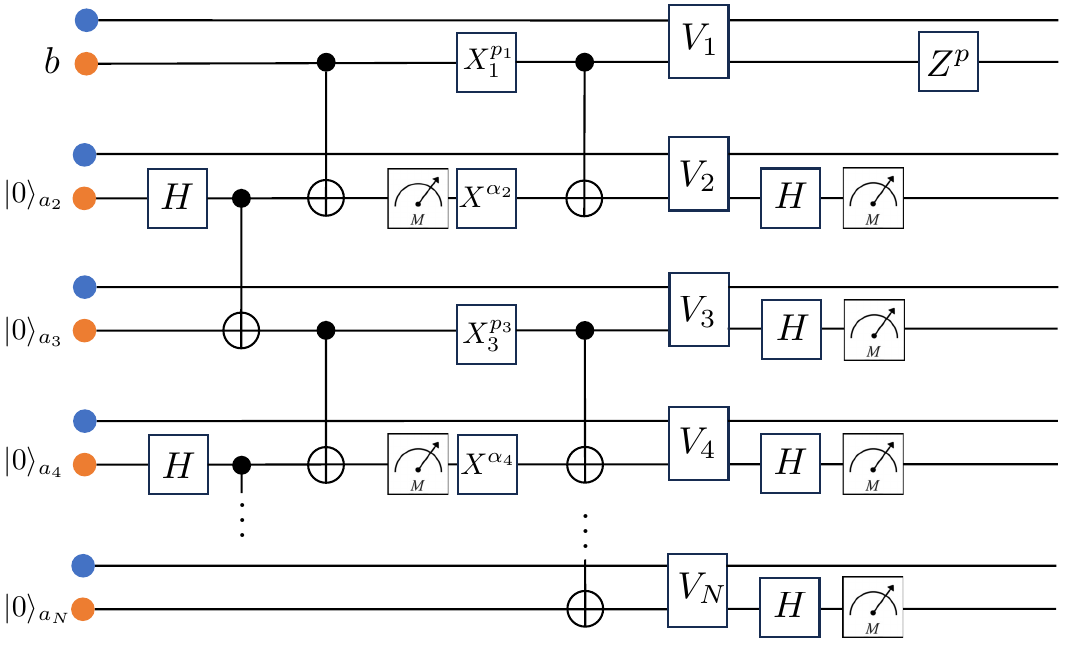}
		\caption{Quantum circuit implementing the unitary in Eq.~\eqref{eq:unitary} (the exchange of bits via classical communication is not shown). Physical and ancillary input qubits are denoted by blue and orange circles, respectively. $X$ and $Z$ are Pauli operators, whose exponents $\alpha_j$, $p_j$ and $p$ are defined in the main text, together with $V_j$. All measurements are in the $Z$-basis.}
		\label{fig:circuit}
	\end{figure}
	
	\noindent Given the (unnormalized) joint input state $\ket{0}_b\ket{\psi_0} + \ket{1}_b\ket{\psi_1}$, the QC implementing $V$ is depicted in Fig.~\ref{fig:circuit} and detailed below. In the first layer, the circuit creates maximally entangled pairs between neighboring ancillas $
	\ket{\Phi^+}_{2j,2j+1}$ ($j=1,2,\dots, N/2-1$). Second,  ${\rm CNOT}_{2j-1,2j}$ gates are applied over pairs of ancillas, except for the last one, $j=1,\ldots ,(N/2-1)$. This layer is followed by a LOCC step: we measure all even ancillas in the $Z$-basis, obtaining measurement outcomes $\alpha_{2j} \in \{0,1\}$, and apply local Pauli corrections $X_{k}^{p_k}$ over all odd ancillas $k\ge 3$, where $p_k = \sum_{2j<k}\alpha_{2j}$. At the same time, the decoupled even ancillas are rotated to the $\ket{0}_{2j}$ state. We then apply another layer of ${\rm CNOT}_{2j-1,2j}$, $j=1,\ldots, N/2$ to all ancilla pairs, yielding the joint state $\ket{0}^{\otimes N}\ket{\psi_0} + \ket{1}^{\otimes N}\ket{\psi_1}$, and proceed by applying to each ancilla and system qubit the control unitary $V_j=|0\rangle \langle 0|\otimes U_{0,j} + |1\rangle\langle
	1|\otimes U_{1,j}$. Finally, we perform a LOCC step: we measure all ancillas except $b$ in the $\ket{\pm}$ basis, yielding the outcomes $\{\beta_j\}_{j=2}^N$ and apply $Z^p_b$ where $p$ is the parity of $\sum_j\beta_j$. This yields $\ket{0}_b U^{(0)}\ket{\psi_0} + \ket{1}_b U^{(1)}\ket{\psi_1}$~\footnote{Implementing the GHZ state as in \cite{piroli2021quantum} yields an alternative protocol for $V$ with $N_a = 2$ and $D = 5$, which may be advantageous for 8-level systems.}. 
	
	\prlsection{Measuring the number of excitations} The unitary~\eqref{eq:unitary} is the key ingredient to our preparation protocol for the Dicke state, as it allows for an efficient measurement of the number of excitations. Consider the state $\ket{\psi}$ and let us define the excitation number $N_e=\sum_{j}n_j$, where $n_j=(1-\sigma^z_j)/2$. Denoting by $\Pi_{j}$ the projector onto the eigenspace of $N_e$ associated with the eigenvalue $j$, we wish to implement the corresponding measurement. It turns out that it is possible to implement a closely related measurement using shallow QCs and LOCC, corresponding to the projectors $\Pi^\ell_{j}=\sum_{i\in \mathcal{T}_{j}^{\ell}}\Pi_{i}$ where $\mathcal{T}_{j}^{\ell}$ is the set of indices $i$ such that $i\equiv j$ (mod $2^\ell$). In particular, we obtain the following:
	\begin{rslt}\label{lem:new_pi}
		The measurement corresponding to the set $\{\Pi^\ell_{j}\}_j$ can be implemented using a circuit with $D=O(\ell)$, $N_a=1$ and $\ell$ additional ancillas.
	\end{rslt}
	
	\begin{figure}
		\includegraphics[scale=0.45]{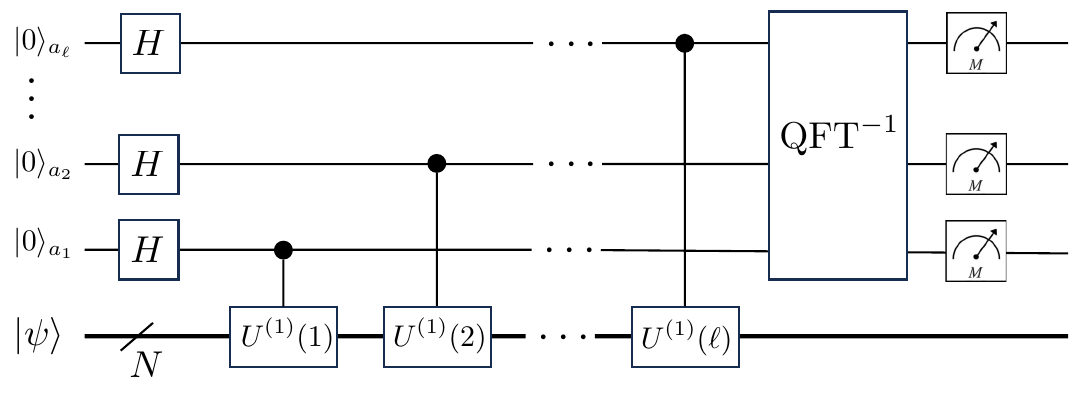}
		\caption{Quantum circuit with $D=O(\ell)$ implementing the measurements corresponding to $\{\Pi_j^\ell\}_j$. The bottom thick line corresponds to the physical Hilbert space of $N$ qubits, while $\ell$ ancillas are attached to the first qubit. Each control-$U$ operation is implemented with depth $O(1)$ via the unitary $V$ in Eq.~\eqref{eq:unitary}. All measurements are performed in the $Z$-basis.}
		\label{fig:measurement}
	\end{figure}

	\noindent The circuit implementing this measurement is represented in Fig.~\ref{fig:measurement}. Attaching all $\ell$ ancillas, initialized in $|0\rangle$, to the first site, the circuit applies to each of them, sequentially, a controlled operator consisting of the unitary operation $V$ in Eq.~\eqref{eq:unitary} with  $U^{(0)}=\openone$ and $U^{(1)}=U^{(1)}(x)=e^{i 2 \pi N_e /2^x}$, where $x=1,\ldots,\ell$ corresponding to each ancilla. At the end of the circuit, an inverse Quantum Fourier Transform (QFT) is applied to the $\ell$ ancillas. This unitary requires depth $D = O(\ell)$~\cite{maslov2007linear} (even assuming 1D locality constrains). It is easy to see that these operations map a state $\ket{\psi}$ into $\sum_{i_1,\ldots, i_\ell=0}^1 |i_1,\ldots,i_\ell\rangle \otimes \Pi^\ell_i |\psi\rangle$, where $i_1\cdots i_\ell$ is the binary representation of $i$. The desired measurement, with the expected probability distribution, is then achieved by performing a projective measurement onto the $\ell$ ancillas. Note that $U^{(1)}(x)=[U^{(1)}(\ell)]^{2^{\ell-x}}$, and thus the protocol is the same of the phase estimation algorithm~\cite{kitaev1995quantum}, with the difference that $\ket{\psi}$ is not an eigenstate for $U^{(1)}(\ell)$. We note that a similar construction to measure the number of excitations was first given in Ref.~\cite{wang2021preparing}.
	
	\prlsection{Preparation of Dicke states} We are now in a position to describe our protocol for the preparation of the Dicke state $\ket{W(M)}$. Fixing $M\leq N/2$, set $p=M/N$ and define $\ket{\Psi(p)}=(\sqrt{1-p}\ket{0}+\sqrt{p}\ket{1})^{\otimes N}$which can be trivially prepared with $D=1$. Now, if we could perform a measurement of the number of excitations and force the outcome to $M$, then we would obtain $\ket{W(M)}$. This is because of the identity
	\begin{equation}\label{eq:psi_p}
		\ket{\Psi(p)}=\sum_{e=0}^N\left[\binom{N}{e}p^e(1-p)^{N-e}\right]^{1/2}\ket{W(e)}\,,
	\end{equation}
	which implies $\Pi_M \ket{\Psi(p)}\propto \ket{W(M)}$. Based on this observation and our previous results, it is easy to devise a preparation scheme. The idea is to perform a measurement corresponding to the projectors $\{\Pi_j^\ell\}_j$ for sufficiently large $\ell$, and repeat the procedure $N_r$ times until we get the desired measurement outcome $M$. At the end of this procedure we obtain a final state $\ket{\psi_\ell}\propto \Pi_M^\ell \ket{\Psi(p)}$. The accuracy of the protocol is controlled by the infidelity $I=|1-|\braket{W(M)|\psi^2_\ell}|^2|$, while the number of repetitions depends on  the probability $P_M$ of obtaining the outcome $M$. By inspection of the state~\eqref{eq:psi_p}, we find $I\sim e^{-2^\ell}$, $P_{M}\sim M^{-1/2}$~\cite{SM}, and we arrive at:
	\begin{rslt}[Preparation of Dicke states]\label{thm_dicke}
		Up to an infidelity $I=\varepsilon$, the Dicke state $\ket{W(M)}$ can be prepared with $N_r=O(\sqrt{M})$, $N_a=1$, $D=O(\ell_{M,\varepsilon})$ and $\ell_{M,\varepsilon}$ additional ancillas, where
		\begin{equation}\label{eq:ell_def_0}
			\ell_{M,\varepsilon}={\rm max}\left\{\log_2(4M), 1+ \log_2 \ln (\sqrt{8\pi M}/\varepsilon) \right\}\,.
		\end{equation} 
	\end{rslt}
	\noindent Alternatively, by slight modifications of the protocol, it is not difficult to show that one can trade the depth with the number of ancillas, realizing a circuit with $N_a=O(\ell_{M,\varepsilon})$, $N_r=O(\sqrt{M})$, $D=O(1)$~\cite{SM}, cf. Table~\ref{table:scalings}. Note that both the number of repetitions and the depth of the circuit do not scale with the system size. In addition, note that we assumed that $\ell_{M,\varepsilon}$ is smaller than $\log_2(L)$. Indeed, for $D=\log_2(L)$ the circuit in Fig.~\ref{fig:measurement} performs a measurement of $N_e$, so the Dicke state is prepared exactly.
	
	We stress that a small infidelity (independent of $N$) automatically guarantees an accurate description of correlation functions. Indeed, denoting by $\braket{\mathcal{O}}_\psi=\braket{\psi|\mathcal{O}|\psi}$, we have $|\braket{\mathcal{O}}_\psi- \braket{\mathcal{O}}_\phi|\leq 2 (1-|\braket{\psi|\phi}|^2)^{1/2}||\mathcal{O}||_{\infty}$, where $||\mathcal{O}||_{\infty}$ is the operator norm. Since the latter equals one for any product of Pauli matrices, we obtain that correlation functions in the prepared state will be arbitrarily close to those of the Dicke state. Finally, in some cases one may need to obtain the Dicke state up to an exponential accuracy. In this case, as mentioned, we can run our protocol implementing the measurement of $N_e$ exactly, leading to an overall depth $O(\log L)$.
	
	\prlsection{The $W$ state} For $M=1$, the above construction gives us an efficient protocol for the $W$ state. In this case, there exists an alternative construction which, while less efficient, is simpler and could be of interest for implementation in NISQ devices. The idea is to prepare the product state $(\sqrt{1-\delta/N} \ket{0}+\sqrt{\delta/N} \ket{1})^{\otimes N}$, and simply measure the parity of the excitations. The protocol is successful if the outcome is odd, in which case it yields a state which we call $\ket{\Phi(\delta)}$. Denoting by $\ket{W}$ the $W$ state, it is easy to see that $|1-|\braket{W|\Phi(\delta)}|^2\leq \delta^2/4$ and that the probability of success is larger than $\delta/2$. On the other hand, the measurement of the parity corresponds to the set $\{\Pi_j^{\ell}\}_j$ with $\ell=1$, so it can be done efficiently using Result~\ref{lem:new_pi}. Therefore, we have the following:
	\begin{rslt}\label{thm_W}
		Up to an infidelity $I=\varepsilon$, the $W$ state can be prepared with $N_r=O(1/\sqrt{\varepsilon})$, $N_a=1$, $D=O(1)$.
	\end{rslt}
	
	\prlsection{Improved scheme via amplitude amplification} 
	Using our previous protocol, the average preparation time of the Dicke state scales as $N_r=O(\sqrt M)$, because we have to do this number of repetitions to have a high probability of success. The reason is that, given the initial state $\ket{\Psi(p)}$, the probability of having $M$ excitations scales as $1/\sqrt M$. We now show how we can exploit the Grover algorithm (or its practical version, named amplitude amplification protocol (AAP)~\cite{brassard1997exact,grover1998quantum,brassard1998quantum,brassard2002quantum}) to improve this result. It is important to notice that a direct application of that algorithm makes the resources dependent on the system size, $N$, something that we want to avoid. Thus, we have to devise an alternative method, which is consistent with the approximation, that circumvents this obstacle.
	
	We recall that, given $\ket{\psi}=\sin\alpha \ket{\psi_1}+\cos\alpha \ket{\psi_2}$, and denoting by $\ket{\tilde \psi}$ the state orthogonal to $\ket{\psi}$ in the subspace generated by $\ket{\psi_1}$ and $\ket{\psi_2}$, the AAP allows one to obtain $\ket{\psi_1}$ by applying a product of $O(1/\alpha)$ unitaries $S_1(\omega_j)$, $S_2(\omega_j)$ (for $\alpha$ small), which act as follows
	\begin{subequations}\label{eq:s1s2}
		\begin{align}
			S_1(\omega)\ket{\psi}&=e^{i\omega}\ket{\psi}\,, \qquad S_1(\omega)\ket{\tilde \psi}=\ket{\tilde \psi}\,,\\
			S_2(\omega)\ket{\psi_1}&=e^{i\omega}\ket{\psi_1}\,, \qquad S_2(\omega)\ket{\psi_2}=\ket{\psi_2}\,,
		\end{align}
	\end{subequations}
	where $\omega_j\in \mathbb{R}$ depend on $\alpha$. Writing $\ket{\Psi(p)} = \sin\alpha \ket{W(M)}+\cos\alpha \ket{R}$, we see that if $S_1(\omega)$, $S_2(\omega)$ can be implemented with circuits of constant depth, then the AAP gives us a deterministic algorithm to obtain $\ket{\psi_1}$ with $D=O(M^{1/4})$, thus reducing the preparation time. 
	
	Realizing the operators in Eqs.~\eqref{eq:s1s2} exactly could be done by known methods using $N_a=\log_2(N)$ and $D=O(1)$~\cite{buhrman2023state}. Instead, we show that, applying ideas similar to those developed so far, an approximate version of them can be realized using a finite amount of resources~\cite{SM}. This leads to the following improved version of Result~\ref{thm_dicke}:
	\begin{rslt}[Improved scheme via amplitude amplification]\label{thm_dicke_improved}
		Up to an infidelity $I=\varepsilon$, the Dicke state $\ket{W(M)}$ can be prepared deterministically ($N_r=1)$, with $N_a=1$, $\ell_{M,\varepsilon}$ additional ancillas, and $D=O(M^{1/4}\ell^2_{M,\varepsilon})$, where 
		\begin{align}\label{eq:ell_def}
			\ell_{M,\varepsilon} = \log_2&\left\{\frac{1}{\ln (4/3)}\left[ 2M(\ln 2M+9/2)\right.\right.\nonumber\\
			+&\left.\left.\ln \left({\rm Poly}(M)/\varepsilon^2\right)	\right]\right\}\,.
		\end{align} 
	\end{rslt}
	
	\prlsection{Eigenstates of the XX Hamiltonian} Going beyond the Dicke model, the previous ideas have ramifications for other Hamiltonians whose eigenstates are labeled by the number of excitations. As a first example, we discuss the well-known XX spin chain $H=-\sum_{k=1}^{N-1} (\sigma^x_k\sigma^x_k+\sigma^y_k\sigma^y_k)$. This model can be solved via the Jordan-Wigner (JW) transformation  $a_k=\left(\prod_{j=1}^{k-1} \sigma^z_j\right)\sigma^-_k$, mapping it to a non-interacting Hamiltonian $H=-\sum_{k=1}^{N-1} (a^\dagger_j a_{j+1}+ {\rm h.c.})$, where $\{a^\dagger_j,a_k\}=\delta_{j,k}$. Accordingly, the eigenstates read $\ket{\Phi(M)}=A^\dagger_M \cdots A^\dagger_1 \ket{0}^{\otimes N}$ with
	\begin{equation}\label{eq:creation}
		A_\alpha=\sum_{k=1}^N c^{\alpha}_k \left(\prod_{j=1}^{k-1} \sigma^z_j\right)\sigma_k^-\,.
	\end{equation}
	Here, $\{c^{\alpha}_k\}$ are distinct sets of coefficients, such that $\{A_\alpha, A_\beta\}=0$,  $\{A^\dagger_\alpha,A_\beta\}=\delta_{\alpha,\beta}$~\cite{franchini2017introduction}, while $M=0,\ldots, N$.
	
	The form of the eigenstates is superficially similar to that of the Dicke states, but it is more complicated due to non-uniform coefficients $c_k^\alpha$ and the string operators $\prod_j \sigma^z_j$. Yet, the anticommutation relations of $A_\alpha$ allows us to devise an efficient preparation protocol. Indeed, the latter implies that $\ket{\Phi(M)}=W_M\cdots W_1 \ket{0}$, where $W_{j}=e^{i\pi(A_j+A_j^\dagger)/2}$. The $W_j$ are unitary and, using our previous constructions, we find that they can be realized deterministically with depth $D=O(N)$~\cite{SM}. Therefore, the eigenstates of the XX Hamiltonian with $M$ excitations can be prepared deterministically ($N_r=1$) by a QC with LOCC of depth $D=O(NM)$ and $N_a=1$. 
	
	The preparation of spin states which can be mapped onto free (or Gaussian) fermionic states has been considered before~\cite{verstraete2009quantum,ferris2014fourier,evenbly2016entanglement,haegeman2018rigorous,kivlichan2018quantum,zhang2018quantum,arute2020observation,babbush2018low,witteveen2021quantum,sopena2022algebraic,ruiz2023bethe}. Ref.~\cite{kivlichan2018quantum} finds a unitary algorithm preparing arbitrary Gaussian operators with depth $O(N)$, yielding a more efficient protocol. However, our approach also allows us to prepare states which are not Gaussian and in principle out of the reach of previous work. For instance, we could prepare $A^\dagger_{\alpha_1}\cdots A^\dagger_{\alpha_n}\ket{\phi_0}$, where $\ket{\phi_0}$ is any linear combination of Gaussian states (assuming $\ket{\phi_0}$ can be prepared efficiently). We also expect that our method could be further improved and generalized to more interesting situations.
	
	\prlsection{Eigenstates of interacting Hamiltonians} As a final example, we consider general states of the form
	\be
	\label{Psiu}
	|\Psi_M\rangle= B_M^\dagger\ldots B_1^\dagger |0\rangle^{\otimes N}\,,
	\ee
	where $B^\dagger_\alpha =\sum_{j=1}^N c_j^\alpha \sigma_j^+$ are interpreted as creating spin excitations. These states are quite general, including the Dicke states and the eigenstates of the so-called Richardson-Gaudin spin chain~\cite{richardson1963restricted,richardson1964exact}, an interacting integrable model. Without assumptions on the coefficients $c_j^\alpha$, efficient preparation of~\eqref{Psiu} is challenging. Here, we will assume that we are in the ``low excitation regime'', namely $M\ll N$, and that
	$\sum_{j=1}^N \bar c_j^\alpha c_j^\beta = \delta_{\alpha,\beta} + O(M/N)$. If $|\psi\rangle$ has at most $M$ excitations, this implies
	\be
	\label{eq:comm}
	[B_\alpha,B^\dagger_\beta]|\psi\rangle = \delta_{\alpha,\beta}|\psi\rangle+ O(M/N)\,.
	\ee
	Namely, $B^\dagger_\alpha$ act as creation operators, up to a $O(M/N)$ error. This allows us to devise a simple preparation protocol, sketched below, and estimate the number of resources needed. We postpone a more detailed analysis of the states~\eqref{Psiu} to future work, including a full study of the Richardson-Gaudin eigenstates.
	
	First, suppose that ~\eqref{eq:comm} holds exactly, \emph{i.e.} without the term $O(M/N)$. Then, we create the state~\eqref{Psiu} by induction. Assuming we have prepared $|\Psi_{M-1}\rangle$, we apply
	\be
	\label{Psie}
	e^{i \theta (B_M+B_M^\dagger)}|\Psi_{M-1}\rangle= \sum_{k=0}^N d_k |\Psi_{M-1+k}\rangle\,,
	\ee
	where we used $B_M|\Psi_{M-1}\rangle=0$, so that the number or excitations cannot decrease. Now, we measure the number of excitations, using the circuit described in Result~\ref{lem:new_pi}, neglecting for simplicity exponentially small errors in the circuit depth. In case we obtain $k=1$ we have succeeded. If $k=0$, we have not changed anything so that we can repeat the procedure. If $k\ge 2$, then we have failed. The probability of failing and obtaining $M=1$ are respectively $O(\theta^4)$ and $O(\theta^2)$. We can iterate this procedure to prepare $\ket{\Psi_M}$ starting from $\ket{0}^{\otimes N}$. It is easy to show that the preparation time scales as $O(M/\theta^2)$, while the success probability as $O(e^{-M\theta^2})$, independent of $N$.
	
	If we do not neglect the term $O(M/N)$ in~\eqref{eq:comm}, then the above construction introduces additional errors. While this is not relevant for the probabilities, the state \eqref{Psie} contains corrections for each $M$. The latter can be estimated as follows. If we have to repeat the procedure $r$ times (on average), the error is $r M/N$ for each step. Accordingly, the total error will be $\varepsilon=r M^2/N$. Since $r=1/\theta^2$, and the probability of not detecting $M=0,1$ in any procedure scales as $p_{\rm fail}=r M \theta^4$, by setting $p_{\rm fail}=1/2$ we have $\varepsilon=M^3/N$. Thus, this allows us to create $M=O(N^{1/3})$ excitations if we take $N$ large.
	
	\prlsection{Outlook} We have introduced protocols to prepare many-body quantum states using QCs and LOCCs. We have shown how we can save resources by relaxing the condition of preparing the states exactly and deterministically, but allowing for controlled infidelities and probabilities of failure. Our results are expected to be relevant for quantum-state preparation in present-day quantum devices, also in light of recent experiments operating QCs assisted by feedforward operations~\cite{baumer2023efficient,chen2023realizing,iqbal2023creation,sukeno2023quantum}. Our work also raises several theoretical questions. For instance, it would be interesting to explore the possibilities of this approach to prepare eigenstates of more general interacting Hamiltonians. In addition, it would be important to understand how the classification of phases of matter via quantum circuits and LOCC introduced in Ref.~\cite{piroli2021quantum} is modified by allowing for finite infidelities. We leave these questions for future work.
	
	\prlsection{Acknowledgments}  We thank Harry Buhrman and Marten Folkertsma for useful discussions. The research is part of the Munich Quantum Valley, which is supported by the Bavarian state government with funds from the Hightech Agenda Bayern Plus. We acknowledge funding from the German Federal Ministry of Education and Research (BMBF) through
	EQUAHUMO (Grant No. 13N16066) within the funding
	program quantum technologies—from basic research to market. This work was funded by the European Union (ERC, QUANTHEM, 101114881). Views and opinions expressed are however those of the author(s) only and do not necessarily reflect those of the European Union or the European Research Council Executive Agency. Neither the European Union nor the granting authority can be held responsible for them.

	
	\let\oldaddcontentsline\addcontentsline
	\renewcommand{\addcontentsline}[3]{}
	\bibliography{bibliography}
	\let\addcontentsline\oldaddcontentsline

	\onecolumngrid
	\newpage

	\appendix
	\setcounter{equation}{0}
	\setcounter{figure}{0}
	\renewcommand{\thetable}{S\arabic{table}}
	\renewcommand{\theequation}{S\thesection.\arabic{equation}}
	\renewcommand{\thefigure}{S\arabic{figure}}
	\setcounter{secnumdepth}{2}

	\begin{center}
		{\Large \bf Supplemental Material}
	\end{center}
	
	Here we provide additional details about the results stated in the main text.
	
	\tableofcontents

	\section{Preparation of the Dicke states}
	\label{app:preparation_dicke}
	In this section we provide additional details on the preparation of the Dicke states. Denoting by $N$ and $M$ the number of qubits and excitations, respectively, and by $\ket{W_N(M)}$ the Dicke state, we provide a full proof for the following statement:
	\begin{prop}\label{prop:dicke_improved}
		For any $\varepsilon>0$, there exists a (non-deterministic) protocol which prepares a state $\ket{\Psi_N}$ with 
		\begin{equation}
			|1-|\braket{W_N(M)|\Psi_N}|^2|\leq\varepsilon\,.
		\end{equation}
		The protocol  is successful with probability
		\begin{equation}
			P_{\rm succ}\geq \frac{1}{\sqrt{8\pi M}}\,,
		\end{equation}
		it uses $N_a=1$ ancilla per site, $D=O(\ell)$ and $\ell$ additional ancillas where
		\begin{equation}\label{eq:definition_ell}
			\ell={\rm max}\left\{\ln(4M)/\ln2, 1+ \frac{\ln \ln (\sqrt{8\pi M}/\varepsilon)}{\ln2} \right\}\,,
		\end{equation}
		independent of $N$. 
	\end{prop}
	\noindent Note that, because the probability of success is $O(1/\sqrt{M})$, the protocol needs to be repeated, on average, $O(\sqrt{M})$ times, which is the result announced in the main text.
	\begin{proof}
		As in the main text, we set $p=M/N$ and start with the initial state (ommitting the dependence on $N$)
		\begin{equation}
			\ket{\Psi(p)}=(\sqrt{1-p}\ket{0}+\sqrt{p}\ket{1})^{\otimes N}=\sum_{e=0}^N\left[\binom{N}{e}p^e(1-p)^{N-e}\right]^{1/2}\ket{W(e)}\,.
		\end{equation}	
		Choose $\ell$ as in~\eqref{eq:definition_ell}, and define $\Pi^\ell_{j}=\sum_{i\in \mathcal{T}_{j}^{\ell}}\Pi_{i}$, where $\Pi_i$ is a projector onto the eigenspace with $i$ excitations, while
		\begin{equation}
			\label{eq:def_taux}
			\mathcal{T}_{j}^{\ell}=\{i: i\equiv j \ ({\rm mod} \ 2^\ell)\}.
		\end{equation}
		We perform a measurement with respect to the projectors $\{\Pi^\ell_j\}$ and repeat the precedure until we obtain the outcome $M$. In case of success, the state reads
		\begin{equation}\label{eq:psi_ell_state}
			\ket{\Psi^{(\ell)}}=\frac{1}{Z_\ell}\sum_{e\in \mathcal{T}_M^{\ell}}\left[\binom{N}{e}p^e(1-p)^{N-e}\right]^{1/2}\ket{W(e)}\,,
		\end{equation}
		where $Z_\ell$ is a normalization factor. According to Result~\ref{lem:new_pi}, this measurement can be performed using a circuit with $D=O(\ell)$, $N_a=1$, and $\ell$ additional ancillas.
		
		We need to estimate the success probability and the distance between $\ket{\Psi^{(\ell)}}$ and $\ket{W_N(M)}$. The former is
		\begin{align}\label{eq:succ_prop_dicke}
			P_{\rm succ}&= Z_\ell^2 = \sum_{e\in \mathcal{T}_M^{\ell}}\left[\binom{N}{e}p^e(1-p)^{N-e}\right]\nonumber\\
			&\geq {N\choose M} p^M (1-p)^{N-M}\geq \frac{1}{2}	\left(\frac{N}{2\pi M(N-M)}\right)^{1/2}\geq \frac{1}{\sqrt{8\pi M}} \,,
		\end{align} 
		where we used
		\begin{equation}\label{eq:stirling_inequality}
			\frac{1}{2}	\left(\frac{N}{2\pi M(N-M)}\right)^{1/2}<{N\choose M} p^M (1-p)^{N-M}<2\left(\frac{N}{2\pi M(N-M)}\right)^{1/2}\,,
		\end{equation}
		which holds for $0<M<N$ and can be proved using known inequalities for the factorial~\cite{robbins1955remark}.
		
		For the overlap, we write
		\begin{equation}\label{eq:useful_id1}
			\sum_{e\in \mathcal{T}_M^\ell}\binom{N}{e}(1-p)^{N-e}p^e\leq {N\choose M} p^M (1-p)^{N-M}+{\rm Pr}[e\leq M-2^\ell]+{\rm Pr}[e\geq M+2^\ell]\,.
		\end{equation}
		Since $2^\ell>M$,  we have ${\rm Pr}[e\leq M-2^\ell]=0$. Let us analyze ${\rm Pr}[e\geq M+2^\ell]$. The Chernoff inequality gives~\cite{arratia1989tutorial}
		\begin{equation}\label{eq:useful_id2}
			{\rm Pr}[e\geq M+2^\ell]\leq \exp\left[-N D\left(\frac{M+2^\ell}{N}\Big\|\frac{M}{N}\right)\right]\,,
		\end{equation}
		where $D(\cdot||\cdot)$ is the relative entropy
		\begin{equation}\label{eq:relative_entropy_def}
			D(a \| p)=a \ln \frac{a}{p}+(1-a) \ln \frac{1-a}{1-p}\,.
		\end{equation}
		We have
		\begin{align}
			\lim_{N\to\infty}N D\left(\frac{M+2^\ell}{N}\Big\|\frac{M}{N}\right)&=-2^\ell + (2^\ell + M) \ln[(2^\ell + M)/M]\,,\\
			\frac{d}{dN} N D\left(\frac{M+2^\ell}{N}\Big\|\frac{M}{N}\right)&<0\,.\label{eq:derivative_3}
		\end{align}
		Eq.~\eqref{eq:derivative_3} implies that $N D\left(\frac{M+2^\ell}{N}\Big\|\frac{M}{N}\right)$ is always larger than its asymptotic value. Therefore
		\begin{equation}
			N D\left(\frac{M+2^\ell}{N}\Big\|\frac{M}{N}\right)\geq  -2^\ell + (2^\ell + M) \ln[(2^\ell + M)/M]\,.
		\end{equation}
		Since $2^{\ell}\geq 4M$, we have
		\begin{equation}
			-2^\ell + (2^\ell + M) \ln[(2^\ell + M)/M]\geq 2^{\ell-1}\,.
		\end{equation}
		Setting $C_{N,M}= {N\choose M} p^M (1-p)^{N-M}$, and putting all together, we get
		\begin{align}
			|\braket{W_N(M)|\Psi^{(\ell)}}|^2&=\frac{C_{N,M}}{Z_\ell^2}\geq 1-\frac{P[e\geq M+2^\ell]}{C_{N,M}}\nonumber\\
			&\geq 1-\sqrt{8\pi M}e^{-2^{\ell-1}}\,,
		\end{align}
		where we used $1/(1+x)\geq 1-x$. 
		
		Finally, using 
		\begin{equation}
			\ell\geq 1+ \frac{\ln \ln (\sqrt{8\pi M}/\varepsilon)}{\ln2}\,.
		\end{equation}
		we get $| |\braket{W_N(M)|\Psi^{(\ell)}}|^2-1|\leq \varepsilon$. Therefore, setting $\ket{\Psi_N}=\ket{\Psi^{(\ell)}}$ we obtain the statement.
	\end{proof}
	
	Next, we prove that the protocol can be slightly modified to trade the depth with the number of ancillas.
	\begin{prop}\label{prop:dicke_parallel}
		For any $\varepsilon>0$, there exists a (non-deterministic) protocol which prepares a state $\ket{\Psi_N}$ with 
		\begin{equation}
			|1-|\braket{W_N(M)|\Psi_N}|^2|\leq\varepsilon\,.
		\end{equation}
		The protocol  is successful with probability
		\begin{equation}
			P_{\rm succ}\geq \frac{1}{\sqrt{8\pi M}}\,,
		\end{equation}
		it uses $D=O(1)$, $N_a=O(\ell)$ ancilla per site, and $\ell$ additional ancillas, where $\ell$ is defined in Eq.~\eqref{eq:definition_ell}.
	\end{prop}
	\begin{proof}
		Compared to the protocol explained in Prop.~\ref{prop:dicke_improved}, we need to reduce the depth of the circuit to $D=O(1)$. To this end, we need to remove the inverse of the quantum Fourier transform (QFT) in the measurement of the excitations (which requires a depth scaling with the number of ancillas) and parallelize the application of the operators $U^{(1)}(x)$. Our parallelization scheme is closely related to the one of Ref.~\cite{hoyer2005quantum}.
		
		We proceed as follows. We define $\ell$ as in Eq.~\eqref{eq:definition_ell}, and append $\ell-1$ ancillas per site, plus $\ell-1$ additional ancillas to the first site (so, in the first site we have $\ell-1+\ell=2\ell-1$ ancillas). All ancillas are initialized in $\ket{0}$. Suppose the initial state of the system is
		\begin{equation}
			\ket{\psi}=\sum_{\{j_k\}}c_{j_1\ldots j_N} \ket{j_1\ldots j_N}\,.
		\end{equation}
		We perform a controlled operations in each local set consisting of one system qubit and $\ell-1$ ancilla qubits, mapping
		\begin{equation}\label{eq:controlled_op}
			\ket{j_k}\otimes \ket{0}^{\otimes (\ell-1)}\mapsto 	\ket{j_k}\otimes \ket{j_k}^{\otimes (\ell-1)}\,,
		\end{equation}
		yielding the state
		\begin{equation}
			\ket{\Psi}=\sum_{\{j_k\}}c_{j_1\ldots j_N} \ket{j_1\ldots j_N}^{\otimes \ell}\,.
		\end{equation}
		The step~\eqref{eq:controlled_op} corresponds to parallel application of fan-out gates and takes constant depth using LOCC~\cite{buhrman2023state}. Next, we apply a Hadamard transformation to each of the $\ell$ ancillas in the first site. Then, for each of the $\ell$ ancillary systems, we apply a unitary $V^{(x)}$ which acts on the $x$-th copy of the system and the $x$-th ancilla in the first site.
		Each unitary is of the form~\eqref{eq:unitary} with  $U^{(0)}=\openone$ and 
		\begin{eqnarray}
			\label{eq:unitary_modified}
			U^{(1)}=U^{(1)}(x)=e^{2i \pi (N_e-M) /2^x}\,,
		\end{eqnarray}
		where $x=1,\ldots,\ell$ corresponding to each ancilla. These operations can be performed in parallel as they act on distinct qubits.  Next, we act with the inverse of \eqref{eq:controlled_op}, apply a Hadamard transformation to each ancilla and measure them in the $Z$-basis. The protocol is successful if we obtain the outcome $0$ for all ancillas. In this case, it is easy to see that the state after the measurement is proportional to the state~\eqref{eq:psi_ell_state}. Note that the unitary~\eqref{eq:unitary_modified} is different from that used in the measurement procedure explained in the main text (Result~\ref{lem:new_pi}). Indeed, while the final state in the case of success is the same as in the previous Proposition, Eq.~\eqref{eq:psi_ell_state}, the outcome is not equal to a projection onto $\Pi_j^\ell$ for different measurement outcomes.
		
		It is immediate to show that the probability of success and the infidelity of the output state are the same as computed in Prop.~\ref{prop:dicke_improved}, which proves the statement.
	\end{proof}
	
	\section{Dicke states from amplitude amplification}
	\label{sec:amplitude_amplification}
	
	In this section we provide additional details for the preparation of the Dicke state using the amplitude-amplification protocol. We start by recalling the precise statement of the latter.
	
	\begin{lem}[Amplitude amplification]\label{lem:amplitude_amplification}
		Let
		\begin{equation}
			\ket{\psi}=\sin\alpha \ket{\psi_1}+\cos\alpha \ket{\psi_2}\,,
		\end{equation}
		where $\ket{\psi_1}$ and $\ket{\psi_2}$ are orthogonal states, and $\ket{\tilde{\psi}}$ be the state orthogonal to $\ket{\psi}$ in the subspace generated by $\ket{\psi_1}$, $\ket{\psi_2}$. Let $	S_1(\omega)$, $	S_2(\omega)$ be two families of unitary operators such that
		\begin{align}
			S_1(\omega)\ket{\psi}&=e^{i\omega}\ket{\psi}\,, \qquad S_1(\omega)\ket{\tilde \psi}=\ket{\tilde \psi}\,,\\
			S_2(\omega)\ket{\psi_1}&=e^{i\omega}\ket{\psi_1}\,, \qquad S_2(\omega)\ket{\psi_2}=\ket{\psi_2}\,,
		\end{align}
		and define
		\begin{equation}\label{eq:def_q}
			Q(\phi,\varphi)=-S_1(\phi)S_2(\varphi)\,.
		\end{equation}
		Then, if the number $m^*=\pi/(4\alpha)-1/2$ is an integer, we have  
		\begin{equation}
			Q^{m^*}(\pi,\pi)\ket{\psi}\propto \ket{\psi_1}\,.
		\end{equation}
		Otherwise, there exist two values $\phi^*, \varphi^*\in \mathbb{R}$ such that
		\begin{equation}
			Q(\phi^*,\varphi^*)Q^{\lfloor m^*\rfloor}(\pi,\pi)\ket{\psi}\propto \ket{\psi_1}\,,
		\end{equation}
		where $\lfloor \cdot \rfloor$ is the integer floor function. 
	\end{lem}
	\begin{proof}
		The proof can be found in Refs.~\cite{brassard1997exact,grover1998quantum,brassard2002quantum}, see in particular Sec. 2.1 in Ref.~\cite{brassard2002quantum}. Note that the lemma states that we can deterministically obtain the state $\ket{\psi_1}$, provided that we can implement the operators $Q(\phi,\varphi)$. They need to be applied a  number of times growing as $\sim 1/\alpha$.
	\end{proof}
	
	Next, we show that the amplitude amplification protocol may be carried out even when the unitaries $S_1(\omega)$ and $S_2(\omega)$ can only be implemented approximately.
	
	\begin{lem}[Approximate amplitude amplification]\label{lem:approximate_amplitude_amplification}
		Let
		\begin{equation}
			\ket{\psi}=\sin\alpha \ket{\psi_1}+\cos\alpha \ket{\psi_2}\,,
		\end{equation}
		where $\ket{\psi_1}$ and $\ket{\psi_2}$ are orthogonal states, and $\ket{\tilde{\psi}}$ be the state orthogonal to $\ket{\psi}$ in the subspace generated by $\ket{\psi_1}$, $\ket{\psi_2}$. Fix $1>\delta>0$ and let $T_1(\omega)$, $T_2(\omega)$ be two families of unitary operators such that
		\begin{subequations}\label{eq:approximate_unitaries}
			\begin{align}
				T_1(\omega)\ket{\psi}&=e^{i\omega}\ket{\psi}\,, \qquad T_1(\omega)\ket{\tilde \psi}=\ket{\tilde \psi}+\varepsilon_1\ket{v}\,,\\
				T_2(\omega)\ket{\psi_1}&=e^{i\omega}\ket{\psi_1}\,, \qquad T_2(\omega)\ket{\psi_2}=\ket{\psi_2}+\varepsilon_2\ket{w}\,,
			\end{align}
		\end{subequations}
		where $0<|\varepsilon_1|,|\varepsilon_2|<\delta/2$, while $\ket{v}$, $\ket{w}$ are normalized states. Finally, set
		\begin{equation}
			P(\phi,\varphi)=-T_1(\phi)T_2(\varphi)\,.
		\end{equation}
		If the number $m^*=\pi/(4\alpha)-1/2$ is an integer, define
		\begin{equation}
			\ket{\chi}=P^{m^*}(\pi,\pi)\ket{\psi}\,,
		\end{equation}
		otherwise, define
		\begin{equation}
			\ket{\chi}= P(\phi^*,\varphi^*)P^{\lfloor m^*\rfloor}(\pi,\pi)\ket{\psi}\,,
		\end{equation}
		where $\lfloor \cdot \rfloor$ is the integer floor function and $\phi^*$, $\varphi^*$ are chosen as in Lemma~\ref{lem:amplitude_amplification}. Then
		\begin{equation}
			||\braket{\psi_1|\chi}|^2-1|\leq  4(\lfloor m^*\rfloor+1)\delta.
		\end{equation}
	\end{lem}
	\begin{proof}
		First, note that
		\begin{align}
			P(\phi,\varphi) \ket{\psi_1}&=-e^{i\varphi}T_1(\phi)\ket{\psi_1}=-e^{i\varphi}T_1(\phi)\left[\zeta_1\ket{\psi} + \zeta_2 \ket{\tilde \psi} \right]=-\zeta_1e^{i\varphi}e^{i\phi}\ket{\psi} - \zeta_2 e^{i\varphi}\ket{\tilde \psi}-\zeta_2e^{i\varphi}\varepsilon_1\ket{v}\nonumber\\
			&=Q(\phi,\varphi)\ket{\psi_1} +\delta_1 \ket{v}\,,
		\end{align}
		where $|\delta_1|\leq \delta/2<\delta$, while $Q$ is defined in~\eqref{eq:def_q}. Here we introduced the coefficients $\zeta_1=\braket{\psi|\psi_1}$, $\zeta_2=\braket{\tilde\psi|\psi_1}$. Similarly, 
		\begin{align}
			P(\phi,\varphi) \ket{\psi_2}&=-T_1(\phi)(\ket{\psi_2}+\varepsilon_2\ket{w})=-T_1(\phi) (\xi_1\ket{\psi}+\xi_2\ket{\tilde \psi})-\varepsilon_2T_1(\phi)\ket{w}\nonumber\\
			&=-\xi_1e^{i\phi}\ket{\psi}-\xi_2\ket{\tilde \psi}-\varepsilon_1\xi_2\ket{v}-\varepsilon_2T_1(\phi)\ket{w}=Q(\phi,\varphi)\ket{\psi_2}+\delta_2\ket{u}\,.
		\end{align}
		Here $\ket{u}$ is a normalized vector, while
		\begin{equation}
			|\delta_2|\leq(|\varepsilon_1\xi_2|^2+|\varepsilon_2|^2+2|\varepsilon_1\varepsilon_2 \xi_2 \braket{v|T_1(\phi)|w}|)^{1/2}\leq \delta\,.
		\end{equation}
		Therefore, we have
		\begin{equation}
			\prod_{j=1}^nP(\phi_j,\varphi_j)\ket{\psi}=\prod_{j=1}^nQ(\phi_j,\varphi_j)\ket{\psi}+\sum_{k=1}^n\left[ c_k\prod_{j=1}^{k-1}P(\phi_j,\varphi_j)\ket{u}+d_k\prod_{j=1}^{k-1}P(\phi_j,\varphi_j)\ket{v}\right]\,,
		\end{equation}
		where $|c_j|\leq \delta$, $|d_j|\leq \delta$. The statement then follows immediately using Lemma~\ref{lem:amplitude_amplification} and that $||\braket{\psi_1|\chi}|^2-1|\leq 2 ||\braket{\psi_1|\chi}|-1|$.
	\end{proof}
	
	We will also use the following result
	\begin{lem}\label{lem:partial_sign}
		Consider the unitary operation defined by
		\begin{equation}\label{eq:partial_sign_flip}
			F_\varphi^{[\ell; m]}\ket{i_1\ldots i_N}=e^{i\varphi f_{\ell, m}(i_1\ldots i_N)}\ket{i_1\ldots i_N}\,,
		\end{equation} 
		where
		\begin{equation}
			f_{\ell, m}(i_1\ldots i_N)=
			\begin{cases}
				1 & {\rm if\ }(\sum_{j=1}^N i_j)-m\equiv 0 \ ({\rm mod}\ 2^\ell)\\
				0 & {\rm otherwise}\,.
			\end{cases}
		\end{equation}
		Then, $F_\varphi^{[\ell; m]}$ can be implemented using $\ell$ total ancillas and a circuit of depth $O(\ell^2)$.
	\end{lem}
	\begin{proof}
		We attach all $\ell$ ancillas to the first site. We prepare them in the state $|+\rangle$, and apply to each of them, sequentially, the unitary operation $V$ in Eq.~\eqref{eq:unitary} with  $U_0=\openone$ and $U_1 =e^{i \pi (N-m) /2^x}$, where $x=1,\ldots,\ell$ corresponding to each ancilla. This can be done by a circuit of depth $O(\ell)$. After that, we apply an inverse QFT to the $\ell$ ancillas, which requires depth $D = O(\ell)$~\cite{maslov2007linear}. This transforms an input state $\ket{\psi}$ as
		\be
		\label{eq:OpforAA}
		W:|+\ldots +\rangle \otimes |\psi\rangle \to \sum_{i_1,\ldots, i_\ell=0}^1 |i_1,\ldots,i_\ell\rangle \otimes \Pi^\ell_{i+m} |\psi\rangle
		\ee
		where $\Pi^\ell_{i+m}$ is the projector onto the subspace with a number of excitations $e$ (namely, a number of values for which $i_k=1$) satisfying $e\equiv i+m$ mod$(2^\ell)$, namely $e-m\equiv i$ mod$(2^\ell)$, and where $i_1\ldots i_\ell$ is the binary decomposition of $i$. We can now apply a unitary to the ancillas in the first site mapping $\ket{0\ldots 0}\mapsto e^{i\varphi}\ket{0\ldots 0}$ and acting as the identity on the other basis states. This operation can be implemented by a local circuit of depth $O(\ell^2)$~\cite{barenco1995elementary}.
		We can finally apply the inverse $W^\dagger$ of the unitary~\eqref{eq:OpforAA}, yielding the desired result. 
	\end{proof}
	
	Finally, we prove our main result of this section.
	
	\begin{prop}[Preparation of Dicke states]
		Let $N\geq 4M$ and $M\geq 1$. For any $0<\delta<1$, there exists an efficient preparation protocol to realize a state $\ket{\Phi}$ such that
		\begin{equation}
			||\braket{\Phi| W_N(M)}|^2-1|\leq 4\delta\,.
		\end{equation}
		The protocol applies a sequence of $2n_M$ unitary operators  which are either $F_\omega^{[\ell,0]}V^\dagger$ or  $F_\omega^{[\ell,M]}$, where $V=e^{-i \theta S_y}$ is a product of local unitaries (which can be implemented in parallel), 
		\begin{align}
			n_M\leq \frac{\pi  (8\pi M)^{1/4}}{2}\,,
		\end{align}
		while 
		\begin{align}\label{eq:def_l_final_prop}
			\ell = \log_2\left\{\frac{1}{\ln (4/3)}\left[ 2M(\ln 2M+9/2)+\ln \left(\frac{{\rm Poly}(M)}{\delta^2}\right)	\right] \right\}\,,
		\end{align}
		with
		\begin{equation}
			{\rm Poly}(M)=\frac{8\pi e^2(8\pi M)^{1/2}}{M\ln (4/3)\left(1-[8/(3\pi M)]^{1/2}\right)}\,.
		\end{equation}
	\end{prop}
	\begin{proof}
		Define
		\begin{equation}
			\ket{\theta}=V\ket{0}^{\otimes N}\,,  
		\end{equation}
		$V=e^{-i \theta S_y}$ with $\cos(\theta)=\sqrt{1-p}$ and
		\begin{equation}
			p=M/N\,.
		\end{equation}
		We start with the identity
		\begin{equation}
			\ket{\theta}=\sin \alpha \ket{W(M)}+\cos\alpha \ket{R}\,.
		\end{equation}
		Here, $\ket{W(M)}$ is the normalized Dicke state with $M$ excitations, while
		\begin{align}
			\ket{R}&=\frac{1}{Z_R}\sum_{e\neq M}\left[\binom{N}{e}p^e(1-p)^{N-e}\right]^{1/2}\ket{W(e)}\,,
		\end{align}
		where $Z_R$ is a normalization factor, and
		\begin{equation}\label{eq:sin_a_def}
			\sin \alpha=\left[{N\choose M} p^M (1-p)^{N-M}\right]^{1/2}\,.
		\end{equation}
		Choose $\ell$ as in~\eqref{eq:def_l_final_prop} and define
		\begin{equation}
			T_1(\omega)=F_\omega^{[\ell,0]}V^\dagger \qquad T_2(\omega) = F_\omega^{[\ell,M]}\,.
		\end{equation}
		Note that $V$ (and hence $V^\dagger$) is a product of local unitaries, while $F_\omega^{[\ell,m]}$ can be implemented efficiently thanks to Lemma~\ref{lem:partial_sign}.  In the following, we will show
		\begin{subequations}\label{eq:to_prove_approx}
			\begin{align}
				T_1(\omega)\ket{\theta}&=e^{i\omega}\ket{\theta}\,, \qquad T_1(\omega)\ket{\tilde \theta}=\ket{\tilde \theta}+\varepsilon_1\ket{v}\,,\label{eq:identity_t1}\\
				T_2(\omega)\ket{W(M)}&=e^{i\omega}\ket{W(M)}\,, \qquad T_2(\omega)\ket{R}=\ket{R}+\varepsilon_2\ket{w}\,,
			\end{align}
		\end{subequations}
		where we denoted by $\ket{\tilde{\theta}}$ the state orthogonal to $\ket{\theta}$ generated by $\ket{W(M)}$ and $\ket{R}$, while $\ket{v}$, $\ket{w}$ are normalized states, with
		\begin{equation}\label{eq:inequality_assumed}
			|\varepsilon_1|, |\varepsilon_2|\leq \frac{\delta}{\pi (8\pi M)^{1/4}}\,.
		\end{equation}
		Combining Lemmas~\ref{lem:amplitude_amplification} and ~\ref{lem:approximate_amplitude_amplification}, we see that this is enough to prove the statement. Indeed, if~\eqref{eq:inequality_assumed} holds, we can implement the approximate amplitude amplification algorithm applying $T_1(\omega)$ and $T_2(\omega)$ a number of times
		\begin{align}
			n_M\leq \frac{\pi}{4\alpha}+\frac{1}{2}\leq \frac{\pi  (8\pi M)^{1/4}}{4}+\frac{1}{2}\leq \frac{\pi  (8\pi M)^{1/4}}{2}\,,
		\end{align}
		where we used $\sin\alpha \leq \alpha$ for $0\leq \alpha \leq 1$, and $(\sin \alpha)^{-1} \leq (8\pi M)^{1/4}$ [which follows from Eqs.~\eqref{eq:stirling_inequality} and \eqref{eq:sin_a_def}]. By Lemma~\ref{lem:approximate_amplitude_amplification}, this gives us the desired state $\ket{W(M)}$ up to an infidelity $I=\varepsilon$ with
		\begin{equation}
			\varepsilon\leq  4\left(\frac{\pi}{4\alpha}+\frac{1}{2}\right) \frac{2\delta}{\pi (8\pi M)^{1/4}}\leq 4\delta\,.
		\end{equation}
		
		Let us prove~\eqref{eq:to_prove_approx}, starting with the action of $T_2(\omega)$. First, it is obvious that $T_2(\omega)\ket{W(M)}=e^{i\omega}\ket{W(M)}$. Next, we have
		\begin{align}
			T_2(\omega)\ket{R}&=\frac{1}{Z_R}T_2(\omega)\left[\sum_{e\in \mathcal{T}_M^{\ell}\setminus\{M\} }\left[\binom{N}{e}p^e(1-p)^{N-e}\right]^{1/2}\ket{W(e)}+\sum_{j\neq 0}\sum_{e\in \mathcal{T}_{M+j}^{\ell} }\left[\binom{N}{e}p^e(1-p)^{N-e}\right]^{1/2}\ket{W(e)}\right]\nonumber\\
			&=\frac{1}{Z_R}\left[e^{i\omega}\sum_{e\in \mathcal{T}_M^{\ell}\setminus\{M\} }\left[\binom{N}{e}p^e(1-p)^{N-e}\right]^{1/2}\ket{W(e)}+\sum_{j\neq 0}\sum_{e\in \mathcal{T}_{M+j}^{\ell} }\left[\binom{N}{e}p^e(1-p)^{N-e}\right]^{1/2}\ket{W(e)}\right]\nonumber\\
			&=\ket{R}+\ket{w}\,.
		\end{align}
		where $\mathcal{T}^{\ell}_M$ is defined in Eq.~\eqref{eq:def_taux}, while
		\begin{equation}
			\ket{w}=\frac{(e^{i\omega}-1)}{Z_R}\sum_{e\in \mathcal{T}_M^{\ell}\setminus\{M\} }\left[\binom{N}{e}p^e(1-p)^{N-e}\right]^{1/2}\ket{W(e)}\,.
		\end{equation}
		We can bound the norm of $\ket{w}$ using
		\begin{equation}
			Z^2_R=1-\binom{N}{M}p^M(1-p)^{N-M}\,,
		\end{equation}
		and the results of Appendix~\ref{app:preparation_dicke}, cf. Eqs.~\eqref{eq:useful_id1},~\eqref{eq:useful_id2}. Using $N\geq 4M$, we obtain 
		\begin{equation}
			\braket{w|w}\leq 4\frac{e^{-2^{\ell-1}}}{Z_R^2}\leq \frac{4e^{-2^{\ell-1}}}{1-\left(\frac{8}{3\pi M}\right)^{1/2}}\leq \frac{\delta^2}{\pi^{2} (8\pi M)^{1/2}}\,,
		\end{equation}
		This inequality holds if
		\begin{equation}
			\ell \geq 1+\log_2\left\{\ln \left[\frac{4\pi^2(8\pi M)^{1/2}}{\delta^2\left(1-\left(\frac{8}{3\pi M}\right)^{1/2}\right)}\right]\right\},
		\end{equation}
		which is true if $\ell$ is chosen as in~\eqref{eq:def_l_final_prop} (this is easily established with the help of numerical inspection).
		
		Next, let us consider $T_1(\omega)$. Again, it is obvious that $T_1(\omega)\ket{\theta}=e^{i\omega} \ket{\theta}$. To prove the second identity in Eq.~\eqref{eq:identity_t1}, we start from
		\begin{equation}
			V^\dagger\ket{\tilde{\theta}}=\frac{1}{Z}V^\dagger (1- \ket{\theta}\bra{\theta})\ket{W(M)}\,,
		\end{equation}
		where
		\begin{equation}
			Z^2=1-|\braket{W(M)|\theta}|^2\,.
		\end{equation}
		It follows from the results of Appendix~\ref{sec:technical_computations} that
		\begin{equation}
			V^\dagger \ket{W(M)}=\sum_{s=0}^Nc_s\ket{W(s)}\,,
		\end{equation}
		cf. Eq.~\eqref{eq:sector_dec}. Therefore, we can write
		\begin{equation}
			V^\dagger\ket{\tilde{\theta}}=\frac{1}{Z}\left(\sum_{s=0}^{2^{\ell}}c_s\ket{W(s)}+\braket{\theta|W(M)}\ket{0}^{\otimes N} \right)+\ket{\tilde{w}}\,,
		\end{equation}
		where $\ket{\tilde{w}}$ has more than $2^{\ell}$ excitations. Using the results of Sec.~\ref{sec:technical_computations} and $N\geq 4M$, we can bound its norm as
		\begin{align}
			\braket{\tilde w|\tilde{w}}&=\frac{1}{Z^2} \sum_{s\geq 2^\ell+1}|c_s|^2\leq \frac{1}{1-[8/(3\pi M)]^{1/2}}\frac{2e^2}{M\pi\ln (4/3)}\exp\left[-2^{\ell}\ln(4/3)+2M(\ln(2M)+9/2)\right]\nonumber\\
			&\leq \frac{\delta^2}{4\pi^2(8\pi M)^{1/2}}\,,
		\end{align}
		where we used that
		\begin{equation}
			2^{\ell}\geq \frac{1}{\ln (4/3)}\left[ 2M(\ln 2M+9/2)+\ln \left(\frac{{\rm Poly}(M)}{\delta^2}\right)	\right]\,,
		\end{equation}
		with
		\begin{equation}
			{\rm Poly}(M)=\frac{8\pi e^2(8\pi M)^{1/2}}{M\ln (4/3)\left(1-[8/(3\pi M)]^{1/2}\right)}\,.
		\end{equation}
		Now, in the space generated by states with at most $2^{\ell}$ excitations, $F^{[\ell,0]}_\omega$ only multiplies the phase $e^{i\omega}$ to the state $\ket{0}^{\otimes N}$, leaving the rest of the basis states invariant. Therefore, we arrive at the final result
		\begin{equation}
			T_1(\omega)]\ket{\tilde \theta}=
			F^{[\ell,0]}_\omega V^\dagger \ket{\tilde{\theta}}= \ket{\tilde{\theta}}+
			\ket{v}\,,
		\end{equation}
		where $\ket{v}=(-\openone + F^{[\ell,0]}_\omega \ket{\tilde w})$, and therefore
		\begin{equation}
			\braket{v|v}\leq \frac{\delta^2}{\pi^2(8\pi M)^{1/2}}\,.
		\end{equation}
	\end{proof}
	
	\section{Technical computations}
	\label{sec:technical_computations}
	
	The goal of this section is to analyze the state
	\begin{equation}
		V^\dagger \ket{W(M)}	
	\end{equation}
	where $\ket{W(M)}$ is the normalized Dicke state with $M$ excitations, while $V=e^{-i \theta S_y}$ with $\cos(\theta)=\sqrt{1-p}$ and $p=M/N$. Throughout this section, we will assume $N\geq 4M$.

	We start by introducing the unnormalized Dicke states
	\begin{align}
		\ket{U(M)}&=\sum_{i_1<\ldots <i_M}\sigma^+_{i_1}\cdots \sigma^+_{i_M}\ket{0}^{\otimes N}\,,
	\end{align}
	and note that we can also write
	\begin{equation}
		\ket{W(M)}=\frac{1}{\sqrt{\binom{N}{M}}}\frac{1}{(N-M)!M!}\sum_{\pi\in S_N}\ket{\underbrace{1\ldots 1}_{M}\underbrace{0\ldots 0}_{N-M}}\,,
	\end{equation}
	where the sum is over all permutations of qubits. Therefore, we can compute
	\begin{align}
		V^\dagger \ket{W(M)}&=\frac{1}{\sqrt{\binom{N}{M}}}\frac{1}{(N-M)!M!}\sum_{\pi\in S_N}e^{+iS_y\theta}\ket{\underbrace{1\ldots 1}_{M}\underbrace{0\ldots 0}_{N-M}}\nonumber\\
		&=\frac{1}{\sqrt{\binom{N}{M}}}\frac{1}{(N-M)!M!}\sum_{\pi\in S_N}(\sqrt{p}\ket{0}+\sqrt{1-p}\ket{1})^{\otimes M}(\sqrt{1-p}\ket{0}-\sqrt{p}\ket{1})^{\otimes (N-M)}
	\end{align}
	We can rewrite
	\begin{align}
		\sum_{\pi\in S_N}&(\sqrt{p}\ket{0}+\sqrt{1-p}\ket{1})^{\otimes M}(\sqrt{1-p}\ket{0}-\sqrt{p}\ket{1})^{\otimes (N-M)}\nonumber\\
		=&\sum_{\pi\in S_N}\left[\sum_{e=0}^M p^{(M-e)/2}(1-p)^{e/2} \ket{U(e)}\right]\left[\sum_{f=0}^{N-M} (1-p)^{(N-M-f)/2}(-p^{1/2})^{f} \ket{U(f)}\right]\nonumber\\
		=&\sum_{e=0}^M \sum_{f=0}^{N-M} p^{(M-e)/2}(1-p)^{e/2}  (1-p)^{(N-M-f)/2}(-p^{1/2})^{f} \sum_{\pi\in S_N}\ket{U(e)} \ket{U(f)}
	\end{align}
	Next, we use
	\begin{align}
		\sum_{\pi\in S_N}\ket{U(e)} \ket{U(f)}&=\binom{M}{e}\binom{N-M}{f}\sum_{\pi\in S_N}\ket{\underbrace{1\ldots 1}_{e+f}\underbrace{0\ldots 0}_{N-e-f}}\nonumber\\
		=&\binom{M}{e}\binom{N-M}{f}\sqrt{\binom{N}{e+f}}(e+f)!(N-e-f)!\ket{W(e+f)}
	\end{align}
	Introducing the variable $s=e+f$, we finally get
	\begin{align}
		V^\dagger \ket{W(M)}&=
		\sum_{s=0}^N\ket{W(s)}\frac{\sqrt{s!(N-s)!}}{\sqrt{M!(N-M)!}}\nonumber\\
		&\times\sum_{e=0}^M \sum_{f=0}^{N-M} \delta_{f+e,s}\binom{M}{e}\binom{N-M}{f}p^{(M-e)/2}(1-p)^{e/2}  (1-p)^{(N-M-f)/2}(-p^{1/2})^{f}\,.
	\end{align}
	Therefore,
	\begin{equation}\label{eq:sector_dec}
		V^\dagger \ket{W(M)}=\sum_{s=0}^Nc_s\ket{W(s)}\,,
	\end{equation}
	where
	\begin{equation}
		c_s=\frac{\sqrt{s!(N-s)!}}{\sqrt{M!(N-M)!}} \times\sum_{e=0}^M\binom{M}{e}\binom{N-M}{s-e}p^{(M-e)/2}(1-p)^{e/2}  (1-p)^{(N-M-(s-e))/2}(-p^{1/2})^{s-e}\,.
	\end{equation}
	Next, we bound $|c_s|$ for $s\geq 3M$. We have
	\begin{equation}
		|c_s|\leq \frac{\sqrt{s!(N-s)!}}{\sqrt{M!(N-M)!}} \sum_{e=0}^M\binom{M}{e}\binom{N-M}{s-e}p^{(M+s-2e)/2} (1-p)^{(N-M-(s-2e))/2}\,.
	\end{equation}
	Using
	\begin{equation}
		\binom{N-M}{s-e}\leq \frac{(N-M)^{s-e}}{(s-M)!}\,,
	\end{equation}
	we obtain
	\begin{align}
		|c_s|&\leq \frac{\sqrt{s!(N-s)!}}{\sqrt{M!(N-M)!}}\frac{1}{(s-M)!}p^{(s-M)/2}\frac{(N-M)^s}{(N-M)^M}(1-p)^{(N-M-s)/2} \sum_{e=0}^M\binom{M}{e}(N-M)^{M-e}p^{M-e} (1-p)^{e}\nonumber\\
		&= \frac{\sqrt{s!(N-s)!}}{\sqrt{M!(N-M)!}}\frac{1}{(s-M)!}p^{(s-M)/2}\frac{(N-M)^s}{(N-M)^M}(1-p)^{(N-M-s)/2} \left[(M+1)\left(1-\frac{M}{N}\right)\right]^M\,.
	\end{align}
	Taking the square, using Stirling's inequality $\sqrt{2 \pi n}\left(\frac{n}{e}\right)^n <n !<\sqrt{2 \pi n}\left(\frac{n}{e}\right)^n e^{\frac{1}{12 n}}$ and rearranging, we arrive at
	\begin{equation}
		|c_s|^2\leq\left(1+\frac{1}{M}\right)^{2M}\frac{1}{\pi}\left(\frac{s(N-s)}{M(N-M)(s-M)^2)}\right)^{1/2}\left(\frac{sM(N-M)}{N-s}\right)^s\left(1-\frac{s}{N}\right)^N \exp\left[2 (s - M) (1 - \ln(s - M)\right]\,.
	\end{equation}
	Now we use that for $s\geq 3M$ one has
	\begin{equation}
		\left(\frac{sM(N-M)}{N-s}\right)^s\left(1-\frac{s}{N}\right)^N \exp\left[2 (s - M) (1 - \ln(s - M)\right]\leq e^{M}	(2M)^{2M}\left[\frac{3(N-M)}{4(N-3M)}\right]^s\,.
	\end{equation}
	This inequality can be established as follows. Denoting the lhs by $g(s)$, we note that the logarithmic derivative $d\ln g(s)/ds$ is a monotonically decreasing function of $s$. Therefore, for $s\geq 3M$, we have $d\ln g(s)/ds\leq d\ln g(s)/ds|_{s=3M}=:\Gamma$. This implies $g(s)\leq g(3M) e^{s\Gamma}$, from which the above inequality follows. Finally, we have
	\begin{equation}
		\left(\frac{s(N-s)}{M(N-M)(s-M)^2)}\right)^{1/2}=	\left[\frac{s}{M}\left(\frac{1}{(s-M)^2}-\frac{1}{(s-M)(N-M)}\right)\right]^{1/2}\leq \frac{2}{M}\,,
	\end{equation}
	and also
	\begin{equation}
		\left[\frac{(N-M)}{(N-3M)}\right]^s=\left(1+\frac{2M}{N-3M}\right)^s\leq \exp\left[\frac{2M s }{N-3M}\right]\leq e^{8M}\,,
	\end{equation}
	where we used $N\geq s$ and $N\geq 4M$. Putting all together, we obtain
	\begin{equation}
		|c_s|^2\leq \frac{2e^2}{\pi M}e^{9M}(2M)^{2M}\left[\frac{3}{4}\right]^s={\rm Poly}(M) \exp\left[-s\ln(4/3)+2M(\ln(2M)+9/2)\right]\,.
	\end{equation}
	Therefore, we arrive at the final result
	\begin{equation}
		\sum_{k\geq s}|c_k|^2\leq \int_{s-1}^{\infty} dk |c_k|^2=\frac{2e^2}{M\pi\ln (4/3)}\exp\left[-(s-1)\ln(4/3)+2M(\ln(2M)+9/2)\right]\,.
	\end{equation}

	\section{Eigenstates of the XX chain}
	We consider the XX Hamiltonian with open boundary conditions
	\begin{equation}\label{eq:xx_hamiltonian}
		H=-\sum_{k=1}^{N-1} (\sigma^x_k\sigma^x_k+\sigma^y_k\sigma^y_k)\,.
	\end{equation}
	Introducing the fermionic modes via the Jordan-Wigner mapping
	\begin{align}
		a_k&=\left(\prod_{j=1}^{k-1} \sigma^z_j\right)\sigma^-_k\,,\qquad 
		a^\dagger_k=\left(\prod_{j=1}^{k-1} \sigma^z_j\right)\sigma^+_k\,,
	\end{align}
	the Hamiltonian~\eqref{eq:xx_hamiltonian} can be rewritten as
	\begin{equation}
		H=-\sum_{k=1}^{N-1} (a^\dagger_j a_{j+1}+ {\rm h.c.})\,.
	\end{equation}
	Note that
	\begin{equation}
		\{a^\dagger_j,a_k\}=\delta_{j,k}\,.
	\end{equation}
	Based on this mapping, a standard result states that the eigenstates of the model read
	\begin{equation}
		\ket{\Psi_M}=A^\dagger_M \cdots A^\dagger_1 \ket{0}^{\otimes N}
	\end{equation}
	with
	\begin{equation}
		A_\alpha=\sum_{k=1}^N c^{\alpha}_k \left(\prod_{j=1}^{k-1} \sigma^z_j\right)\sigma_k^-\,,
	\end{equation}
	where $\{c^{\alpha}_k\}$ are pairwise distinct sets of numerical coefficients, $\{c^\alpha_k\}\neq \{c^\beta_{k}\}$. These operators satisfy the canonical anticommutation relations
	\begin{equation}\label{eq:anticommutations}
		\{A_\alpha, A_\beta\}=0\,, \qquad \{A^\dagger_\alpha,A_\beta\}=\delta_{\alpha,\beta}\,,
	\end{equation}
	and so one has the constraint
	\begin{equation}
		\sum_{j=1}^N\overline{c_{j}^\alpha}c_j^\beta=\delta_{\alpha,\beta}\,.
	\end{equation}
	We prove the following statement
	\begin{prop}
		The eigenstates  $\ket{\Psi_M}$ can be prepared by a circuit (with LOCC) of depth $O(NM)$ and a single ancilla per site.
	\end{prop}
	\begin{proof}
		First, note that
		\be 
		|\Psi_M\rangle = e^{i\frac{\pi}{2} (A_M+ A_M^\dagger)} |\Psi_{M-1}\rangle\,.
		\ee
		This is because $(A_M+A_M^\dagger)^2=1$ and also $A_M|\Psi_{M-1}\rangle=0$. The last equality follows from the anticommutation relations~\eqref{eq:anticommutations} and the fact that $A_M\neq A_j$ for $j<M$. 
		
		Next, let $A$ and $B$ two anticommuting operators, such that $\{A,A^\dagger\}=\{B,B^\dagger\}=1$. Then, for $\alpha,\theta\in \mathbb{R}$, we have the identity
		\begin{eqnarray}
			\label{eq:aux}
			e^{i \alpha [\cos(\theta)(A+A^\dagger)+\sin(\theta) (B+B^\dagger)]} =e^{i \beta (B+B^\dagger)} e^{i \gamma (A+A^\dagger)}e^{i \beta (B+B^\dagger)}\,,
		\end{eqnarray}
		where 
		\begin{align}
			\cos(2\beta)\cos(\gamma)=\cos(\alpha)\,,\\
			\sin(2\beta)\cos(\gamma)=\sin(\alpha)\sin(\theta)\,,\\
			\sin(\gamma)=\sin(\alpha)\cos(\theta)\,,
		\end{align}
		which can be simply derived expanding the exponential functions. Note that the third equation can be derived from the first two, so that there is always a solution to the above system. By applying iteratively this relation, we obtain
		\be 
		e^{i\theta (A_M+ A_M^\dagger)} = R_N \ldots R_2 R_1 R_2 \ldots R_N,
		\ee
		where $R_j = e^{i\theta_j (c_j^M a_j^\dagger + {\rm h.c.}) }$, where $\theta_j$ can be easily computed through iteration of \eqref{eq:aux}. Note that the coefficients $c_j^M$ are in general complex, but Eq.~\eqref{eq:aux} can be applied as the phases are reabsorbed in the definition of the operators $A$ and $B$. 
		
		Finally, we notice that 
		\be
		R_j = V_j X_j V^\dagger_j\,,
		\ee
		where, $X_j=e^{i\theta_j (c_j^M \sigma_j^+ +{\rm h.c.})}$ and
		\be
		\label{eq:unitary_JW}
		V_j = |0\rangle_j\langle 0| \otimes \left(\otimes_{k=1}^{j-1} U_{0,k}\right)+ |1\rangle_j\langle 1| \otimes \left(\otimes_{k=1}^{j-1} U_{1,k}\right)\,,
		\ee
		with $U_{0,k}=1$ and $U_{1,k}=\sigma_z$. It is easy to see that it performs the Jordan Winger transformation, \emph{i.e.} 
		\begin{equation}
			V_j \sigma_j^+ V^\dagger_j =\sigma_{1}^z\otimes \cdots \otimes \sigma_{j-1}^z\otimes \sigma_j^+\,.
		\end{equation}
		Note that Eq.~\eqref{eq:unitary_JW} is of the form~\eqref{eq:unitary}, where the control qubit is at site $j$, and can thus be implemented with a circuit of depth $D=O(1)$ using LOCC. 
		
		Finally, using
		\begin{align}
			X_N\ldots X_{j+1} V_{j}&=  V_{j} X_N\ldots X_{j+1}\,,\\
			V_{j} X_{j+1} \ldots X_N &=  X_{j+1} \ldots X_N  V_{j} \,,
		\end{align}
		and that $V_j=V^\dagger_j$, $[V_j,V_k]=0$, we obtain
		\begin{equation}
			e^{i\theta (A_M+ A_M^\dagger)} = T L_N L_{N-1} \ldots  L_{1} \tilde L_2 \ldots, \tilde L_N T,\\
		\end{equation}
		where $T =  \left[\prod_{j=1}^N V_j\right]$, and
		\begin{align}
			L_N&= X_N, \quad L_{N-1}=V_N X_{N-1},\ldots, L_{2}=V_3 X_{2}, \quad 	L_1=V_2 X_1 V_2, \\
			\tilde L_N &= X_N, \quad \tilde L_{N-1}=X_{N-1}V_N ,\ldots, \tilde L_{2}= X_{2} V_3 \,.
		\end{align}
		Therefore, putting together all the excitations, and denoting by ${\cal L}_k = L_N L_{N-1} \ldots  L_{1} \tilde L_2 \ldots, \tilde L_N$ the operator corresponding  to the $k$-th excitation (\emph{i.e.}, depending on the parameters $c^k_j$), we have
		\be
		|\Psi_M\rangle = T {\cal L}_M \ldots {\cal L}_1 |0\ldots 0\rangle.
		\ee
		Since each $L_j$ and $V_j$ can be implemented by a circuit of depth $D=O(1)$, we immediately obtain the statement.
	\end{proof}
	
\end{document}